\documentclass[11pt]{article}

\usepackage[letterpaper,margin=1in]{geometry}

\usepackage[sorting=nty]{biblatex}
\usepackage{amsfonts}
\usepackage{amsmath}
\usepackage{amssymb}

\usepackage{amsthm}
\newtheorem{definition}{Definition}[section]
\newtheorem{remark}{Remark}[section]
\newtheorem{fact}{Fact}[section]
\newtheorem{proposition}{Proposition}[section]
\newtheorem{theorem}{Theorem}[section]
\newtheorem{corollary}{Corollary}[theorem]
\newtheorem{lemma}{Lemma}[section]

\usepackage{mathtools}
\usepackage{enumerate}
\usepackage{graphicx} 
\usepackage[linesnumbered,commentsnumbered,ruled,noend]{algorithm2e}

\usepackage{url}
\usepackage{hyperref}
\usepackage[hyphenbreaks]{breakurl}
\usepackage{cleveref}

\newcommand{\floor}[1]{\left\lfloor #1 \right\rfloor}
\newcommand{\ceil}[1]{\left\lceil #1 \right\rceil}

\newcommand{\cbr}[1]{\left\{ #1 \right\}}
\newcommand{\rbr}[1]{\left( #1 \right)}
\newcommand{\sbr}[1]{\left[ #1 \right]}

\newcommand{\bitand}{\;\&\;}
\newcommand{\shr}{\;>>\;}

\newcommand{\N}{\mathbb{N}}

\newcommand{\U}[1]{\mathbb{U}_{#1}}

\newcommand{\myemail}[1]{\href{mailto:#1}{\texttt{#1}}}

\providecommand{\keywords}[1]{{\noindent\small\textbf{\textit{Keywords---}} #1}}

\title{Classic Round-Up Variant of Fast Unsigned Division by Constants: Algorithm and Full Proof}
\author{Yifei Li\thanks{Email address: \myemail{liyifei.411@outlook.com}.}}
\date{}

\begin{document}

\maketitle


\begin{abstract}
    Integer division instruction is generally expensive in most architectures.
    If the divisor is constant, the division can be transformed into combinations of several inexpensive integer instructions.
    This article discusses the classic round-up variant of the fast unsigned division by constants algorithm, and provides full proof of its correctness and feasibility.
    Additionally, a simpler variant for bounded dividends is presented.
\end{abstract}



\section{Introduction}

Processors implement instructions with pipelines to leverage instruction-level parallelism.
Different instructions require different numbers of pipeline stages, and thus take different clock cycles to finish.
Throughout the years, integer division has been one of the most time-consuming instructions in most architectures, for it not only requires many stages, but also cannot be fully pipelined~\cite{granlund_division_1994,hennessey_architecture_2019}.

Fortunately, if the divisor is a known constant, the costly division instruction can be replaced with combinations of inexpensive integer instructions, including addition, multiplication, and bit shifting.
Torbj\"{o}rn Granlund and Peter L. Montgomery proposed the first algorithms for integer division by arbitrary non-zero constants using integer multiplication with a pre-computed magic number~\cite{granlund_division_1994}.
The unsigned version is known as the \emph{round-up} method in the later literature, because it computes the integer magic number by rounding up the exact value of the magic number which is possibly a floating point.
They also discovered that in some cases the magic number they found can be unnecessarily large, and proposed an improved version for such cases.

The book \textit{Hacker's Delight} includes a family of such round-up variants of fast integer division by constants, and provides general algorithms to find the minimum possible magic number~\cite{warren_hacker_2012}.
As a duality of the round-up method, Arch D. Robison proposed the \emph{round-down} method which computes the integer magic number by rounding down the exact value of the magic number~\cite{robison_division_2005}.

When constructing a compiler, the round-down variant and the round-up variant with minimum magic number are preferred, because they spend more time to find the magic number in compile time and execute fewer instructions for division in runtime.
However, the classic round-up variant with potentially larger but easy-to-find magic numbers is still commonly used in other scenarios where finding the magic number itself is part of the overall runtime.
One example is the heterogeneous computation, where the magic number is pre-computed on the host and the division is executed on the device, while both procedures are included in the runtime.
This article focuses on the classic round-up variant, especially the formal proof of this algorithm.

The proof by Granlund and Montgomery is brief and omits some details~\cite{granlund_division_1994}.
Moreover, they mentioned overflow handling, but did not give the condition of overflow.
The proof in \textit{Hacker's Delight} only covers the variant with minimum magic number~\cite{warren_hacker_2012}.
A blog by Peter Ammon informally discusses the round-up variant with more intuition and explanation~\cite{ammon_division_2010}.
It provides many inspiring ideas, but is ambiguous at several critical points.
A more recent blog by Ruben van Nieuwpoort presents formal and rigorous proofs of the round-up and round-down methods~\cite{nieuwpoort_division_2020}.
However, it focuses more on the round-down variant.
As for the round-up variant, it mainly follows from the proof by Granlund and Montgomery, which somehow lacks intuition compared with Ammon's work.

This article first presents the classic round-up variant of the fast unsigned division by constants algorithm, shown in \Cref{alg:pre} and \Cref{alg:runtime}, which does not require finding the minimum possible magic number.
Then it formulates the problem with theorems and propositions in \Cref{sec:theory}, and provides rigorous yet intuitive proofs of the algorithm in \Cref{sec:proof-main,sec:proof-no-overflow}.
In addition, this article discusses an even simpler variant as shown in \Cref{alg:faster-runtime} with fewer instructions in the runtime phase.
As pointed out in the prior art~\cite{granlund_division_1994,ammon_division_2010,nieuwpoort_division_2020}, it is not feasible for a general $N$-bit unsigned division due to arithmetic overflow.
However, this article shows that if the dividend is bounded and strictly less than $2^{N-1}$, this simpler variant will work correctly without overflow.

\section{Problem Statement}

Let a constant $N$-bit unsigned integer $d>0$ be the divisor.
Given any $N$-bit unsigned integer $n$ as dividend, we want to compute the round-down quotient $\floor{\frac{n}{d}}$ as the output, which is also an $N$-bit unsigned integer.


\section{Algorithms}
\label{sec:algo}


\begin{algorithm}[ht]
    \caption{Fast unsigned division: magic number pre-computation\label{alg:pre}}

    \small
    \SetAlgoLined
    \SetKwComment{Comment}{// }{}

    \SetKwInOut{Input}{Input}
    \SetKwInOut{Output}{Output}

    \Input{$N$-bit unsigned integer $d>0$}
    \Output{$N$-bit unsigned magic number $m_{lo}$, and shift offset $p$}

    $p\gets \ceil{\log_2{d}}$\;
    $k\gets N + p$\;
    $m\gets \ceil{2^k / d}$\;
    $m_{lo}\gets m \bitand \rbr{2^{N} - 1}$\;
\end{algorithm}


\begin{algorithm}[ht]
    \caption{Fast unsigned division: runtime\label{alg:runtime}}

    \small
    \SetAlgoLined
    \SetKwComment{Comment}{    // }{}

    \SetKwInOut{Input}{Input}
    \SetKwInOut{Output}{Output}

    \Input{$N$-bit unsigned integer $n$, $N$-bit unsigned magic number $m_{lo}$, and shift offset $p$}
    \Output{$N$-bit unsigned round-down quotient $t=\floor{\frac{n}{d}}$}

    $q\gets \rbr{m_{lo}\cdot n} \shr N$\;
    $h\gets \min(p,\; 1)$ \Comment*[l]{h = 0 iff d = 1}
    \label{alg:runtime:min}
    $t\gets \rbr{n - q} \shr h$\;
    $t\gets t + q$\;
    $t\gets t \shr \rbr{p - h}$\;
\end{algorithm}


Given a constant unsigned divisor $d>0$, pre-compute the $N$-bit unsigned magic number $m_{lo}$ and the shift offset $p$ with \Cref{alg:pre}.
At runtime, for any $N$-bit unsigned dividend $n$, \Cref{alg:runtime} computes the round-down quotient $t$.
Note that $h\gets\min(p, 1)$ handles the case where $d=1$ and the division should return $n$ directly without any bit shifting.
In other cases, $h$ always equals $1$.


\begin{algorithm}[ht]
    \caption{Fast unsigned division: faster runtime for bounded dividend\label{alg:faster-runtime}}

    \small
    \SetAlgoLined
    \SetKwComment{Comment}{// }{}

    \SetKwInOut{Input}{Input}
    \SetKwInOut{Output}{Output}

    \Input{$N$-bit unsigned integer $n<2^{N-1}$, $N$-bit unsigned magic number $m_{lo}$, and shift offset $p$}
    \Output{$N$-bit unsigned round-down quotient $t=\floor{\frac{n}{d}}$}

    $q\gets \rbr{m_{lo}\cdot n} \shr N$\;
    $t\gets \rbr{n + q} \shr p$\;
\end{algorithm}


Furthermore, if it is ensured that $n<2^{N-1}$, i.e., the most significant bit of $n$ is 0, then \Cref{alg:faster-runtime} computes the quotient $t$ with fewer instructions.
However, this algorithm can encounter an arithmetic overflow if the dividend is out of the bound.

In practice, $N$ often equals 32 and the algorithm deals with 32-bit unsigned integers.
Most architectures provide the instruction that returns the high 32-bit word of unsigned multiplication, which is useful for computing $\rbr{m_{lo}\cdot n} >> 32$.


\section{Correctness and Feasibility}
\label{sec:theory}

\begin{definition}
    Given $N\in\N$ such that $N>0$, $u\in\N$ is called an $N$-bit unsigned integer if $0 \le u < 2^{N}$.
    Let $\U{N}$ denote the set of all $N$-bit unsigned integers, i.e., $\U{N} \coloneqq \cbr{0, 1, \dots, 2^N-1}$.
    Let $\U{N}^+ \coloneqq \U{N} \setminus \cbr{0}$, i.e., the non-zero $N$-bit unsigned integers.
    \label{def:uint}
\end{definition}

\begin{definition}[shift offset]
    For divisor $d \in\U{N}^+$, its \emph{shift offset} $p$ is defined as
    \begin{equation}
        p\coloneqq \ceil{\log_2{d}}.
        \label{def:shift}
    \end{equation}
\end{definition}

\begin{definition}[magic number]
    For divisor $d \in\U{N}^+$, its $N$-bit unsigned \emph{magic number} $m_{lo}$ is defined as
    \begin{equation}
        m_{lo}\coloneqq \ceil{\frac{2^{N + p}}{d}} \bitand \rbr{2^{N} - 1},
    \end{equation}
    where $p$ is the shift offset.
    \label{def:magic}
\end{definition}

\begin{fact}
    $\forall u\in\U{N}$, and $M\in\N$ such that $0 \le M \le N$, it satisfies that
    \begin{equation}
        \floor{\frac{u}{2^{M}}} = u \shr M.
        \label{eq:floor-shift}
    \end{equation}
\end{fact}

\begin{theorem}[\Cref{alg:faster-runtime} correctness]
    Given $d \in\U{N}^+$ as divisor, and $n\in\U{N}$ as dividend, let
    \begin{equation}
        q\coloneqq \floor{\frac{m_{lo}\cdot n}{2^{N}}},
        \label{eq:q}
    \end{equation}
    where $m_{lo}$ is the magic number.
    It satisfies that
    \begin{equation}
        \floor{\frac{n}{d}}=\floor{\frac{n + q}{2^p}},
        \label{eq:main}
    \end{equation}
    where $p$ is the shift offset.
    \label{thm:main}
\end{theorem}

\begin{proposition}
    In \cref{eq:main}, the addition $n + q$ can overflow an $N$-bit register.
    \label{prop:n+q}
\end{proposition}

\begin{proposition}[partial feasibility]
    If $n<2^{N-1}$, RHS of \cref{eq:main} is feasible for $N$-bit registers.
    \label{prop:faster-feasible}
\end{proposition}

\begin{corollary}[\Cref{alg:runtime} correctness]
    Given $d \in\U{N}^+$ such that $d\not=1$ as divisor, and $n\in\U{N}$ as dividend, it satisfies that
    \begin{equation}
        \floor{\frac{n}{d}} = \floor{\frac{\floor{\rbr{n - q}/\;2} + q}{2^{p-1}}},
        \label{eq:main-no-overflow}
    \end{equation}
    where $q$ is defined by \eqref{eq:q}, and $p$ is the shift offset.
    \label{thm:no-overflow}
\end{corollary}

\begin{remark}
    It is easy to verify that when $d=1$ \Cref{alg:runtime} is also correct, because the trick in \cref{alg:runtime:min} ensures that no bit shifting is actually performed.
\end{remark}

\begin{proposition}[full feasibility]
    RHS of \cref{eq:main-no-overflow} is feasible for $N$-bit registers.
    \label{prop:feasible}
\end{proposition}


\section{Proof of Theorem \ref{thm:main}}
\label{sec:proof-main}

\begin{definition}[full magic number]
    For divisor $d \in\U{N}^+$, its \emph{full magic number} $m$ is defined as
    \begin{equation}
        m\coloneqq \ceil{\frac{2^{N + p}}{d}},
    \end{equation}
    where $p$ is the shift offset.
\end{definition}

\begin{remark}
    For divisor $d \in\U{N}^+$, its magic number $m_{lo}$ is the lower $N$ bits of its full magic number $m$.
\end{remark}

\begin{remark}
    For divisor $d \in\U{N}^+$ such that $d$ is a power of 2, we have $m=2^N$ and $m_{lo}=0$.
    \label{rem:d-pow-2}
\end{remark}

\begin{lemma}
    $\forall u \in\U{N}^+$, it satisfies that
    \begin{equation}
        \frac{2^{\ceil{\log_2{u}}}}{u} \le \frac{2^N}{2^{N-1}+1},
    \end{equation}
    and the equality holds iff $u=2^{N-1}+1$.
    \label{lem:ratio}
\end{lemma}

\begin{proof}
    We know $\rbr{2^{\ceil{\log_2{u}}}/\;u} = 2^{\ceil{\log_2{u}} - \log_2{u}}$.
    When $u$ is a power of 2, $\rbr{2^{\ceil{\log_2{u}}}/\;u}$ takes its minimum value $1$, because $\ceil{\log_2{u}} - \log_2{u} \ge 0$ and the equality holds iff $u$ is a power of 2.

    When $u$ is not a power of 2, suppose $2^{M} < u < 2^{M+1}$ for some $M\in\N$ such that $0<M< N$.
    It is easy to see $\forall u\in\cbr{2^{M} + 1, \dots, 2^{M+1}-1}$, $\ceil{\log_2{u}} = M+1$.
    To maximize $\rbr{2^{\ceil{\log_2{u}}}/\;u}$, $u$ should take the minimum value in this set, i.e., $2^{M} + 1$. Thus,
    \begin{equation}
        \max_{u\in\cbr{2^{M} + 1, \dots, 2^{M+1}-1}}\frac{2^{\ceil{\log_2{u}}}}{u} = \frac{2^{M+1}}{2^{M}+1}.
    \end{equation}

    Further, we can prove that $\forall x\in \mathbb{R}, f(x)=2^{x+1} / \rbr{2^{x}+1}$ increases monotonically.
    So when $M=N-1$, i.e., $u=2^{N-1}+1$, $\rbr{2^{\ceil{\log_2{u}}}/\;u}$ is maximized over $u$, and
    \begin{equation}
        \max_{u \in\U{N}^+}\frac{2^{\ceil{\log_2{u}}}}{u} = \frac{2^N}{2^{N-1}+1}.
    \end{equation}
\end{proof}

\begin{lemma}
    Given $d\in\U{N}^+$ as divisor, its full magic number $m$ is an unsigned integer of \emph{exactly} $N+1$ bits, i.e., $m\in\U{N+1}\setminus\U{N}$.
    \label{lem:33bits}
\end{lemma}

\begin{proof}
    On the one hand,
    \begin{equation}
        m = \ceil{\frac{2^{N + p}}{d}}
        \ge {\frac{2^{N + p}}{d}}
        = {\frac{2^{N + \ceil{\log_2{d}}}}{2^{\log_2{d}}}}
        \ge {\frac{2^{N + \ceil{\log_2{d}}}}{2^{\ceil{\log_2{d}}}}}
            = 2^{N}.
    \end{equation}
    On the other hand, by \Cref{lem:ratio} we have
    \begin{equation}
        \frac{2^p}{d} = \frac{2^{\ceil{\log_2{d}}}}{d} \le \frac{2^{N}}{2^{N-1}+1} \le \frac{2^{N} - 1}{2^{N-1}}.
    \end{equation}
    Multiplying both sides by $2^N$ we get
    \begin{equation}
        \frac{2^{N + p}}{d} \le 2^{N+1} - 2 .
    \end{equation}
    Rounding up both sides we have
    \begin{equation}
        \ceil{\frac{2^{N + p}}{d}} \le 2^{N+1} - 2 < 2^{N+1}.
    \end{equation}
    So $2^{N} \le m < 2^{N+1}$.
    By definition, $m\in\U{N+1}\setminus\U{N}$, i.e., $m$ takes exactly $N+1$ bits.
\end{proof}

\begin{remark}
    By \Cref{lem:33bits}, it is easy to see that
    \begin{equation}
        m_{lo} = m - 2^{N}.
        \label{eq:m_lo}
    \end{equation}
\end{remark}

\begin{lemma}
    Given $d \in\U{N}^+$ as divisor, and $n\in\U{N}$ as dividend, it satisfies that
    \begin{equation}
        \floor{\frac{n}{d}}=\floor{\frac{mn}{2^{N+p}}},
        \label{eq:lemma5_2}
    \end{equation}
    where $p$ is the shift offset.
    \label{lem:mn_version}
\end{lemma}

\begin{proof}
    For simplicity, let $k \coloneqq N+p$.
    Then we have
    \begin{equation}
        m = \ceil{\frac{2^{k}}{d}} = \frac{2^{k} + e}{d},
    \end{equation}
    where $e\in\N$ is defined as
    \begin{equation}
        e\coloneqq d - 1 - \sbr{(2^{k} - 1)\;\mathrm{mod}\;d}.
    \end{equation}
    RHS of \eqref{eq:lemma5_2} becomes
    \begin{equation}
        \floor{\frac{mn}{2^{k}}}
        = \floor{\frac{2^{k} + e}{d}\cdot\frac{n}{2^{k}}}
        = \floor{\frac{n}{d} + \frac{e}{d}\cdot\frac{n}{2^{k}}}.
    \end{equation}
    To prove \eqref{eq:lemma5_2}, we only need to show that
    \begin{align}
        \floor{\frac{n}{d} + \frac{e}{d}\cdot\frac{n}{2^{k}}} = \floor{\frac{n}{d}}
         & \iff \frac{n}{d} + \frac{e}{d}\cdot\frac{n}{2^{k}} < \floor{\frac{n}{d}} + 1        \\
         & \iff \rbr{\frac{n}{d} - \floor{\frac{n}{d}}} + \frac{e}{d}\cdot\frac{n}{2^{k}} < 1.
        \label{eq:lemma5_2_target}
    \end{align}
    Note that the fractional part of $\frac{n}{d}$ is at most $\frac{d-1}{d}$.
    It is easy to see that $0 \le e < d$, so $0\le\frac{e}{d}<1$.
    Besides, by definition $n<2^N$, and we have
    \begin{align}
        \frac{n}{2^k} = \frac{n}{2^{N+p}} < \frac{2^N}{2^{N+p}}
        = \frac{1}{2^p} = \frac{1}{2^{\ceil{\log_2{d}}}}
        \le \frac{1}{2^{\log_2{d}}} = \frac{1}{d}.
    \end{align}
    Plugging these into LHS of \eqref{eq:lemma5_2_target}, we have
    \begin{equation}
        \rbr{\frac{n}{d} - \floor{\frac{n}{d}}} + \frac{e}{d}\cdot\frac{n}{2^{k}} < \frac{d-1}{d} + 1\cdot \frac{1}{d} = 1.
    \end{equation}
    Thus, we prove \eqref{eq:lemma5_2}.
\end{proof}

\begin{remark}
    In \cref{eq:lemma5_2}, we know that $m$ has exactly $N+1$ bits, and $n$ has up to $N$ bits. If $N>1$, then by \Cref{thm:width}, $mn$ has up to $2N+1$ bits.
    \label{rem:65bits}
\end{remark}

Now, we can prove \Cref{thm:main} as follows.
\begin{proof}
    By \Cref{lem:mn_version}, in order to prove \eqref{eq:main}, we only need to show that
    \begin{equation}
        \floor{\frac{mn}{2^{N+p}}}=\floor{\frac{n + q}{2^p}}.
        \label{eq:main_target}
    \end{equation}
    By \eqref{eq:m_lo}, $mn$ can be rewritten as
    \begin{equation}
        mn = \rbr{m_{lo} + 2^{N}}n
        = m_{lo}\cdot n + 2^{N}n.
        \label{eq:mn-rewrite}
    \end{equation}
    LHS of \eqref{eq:main_target} becomes
    \begin{equation}
        \floor{\frac{mn}{2^{N+p}}}
        = \floor{\frac{m_{lo}\cdot n + 2^{N}n}{2^{N+p}}}
        = \floor{\frac{\rbr{m_{lo}\cdot n} /\; 2^{N} + n}{2^{p}}}.
        \label{eq:main_target_lhs}
    \end{equation}
    Note that $\forall u,M\in\N$, it satisfies that
    \begin{equation}
        \frac{u}{2^{M}} \le \floor{\frac{u}{2^{M}}} + \frac{2^{M} - 1}{2^{M}} < \floor{\frac{u}{2^{M}}} + 1.
    \end{equation}
    So we have
    \begin{equation}
        \frac{m_{lo}\cdot n}{2^{N}} < \floor{\frac{m_{lo}\cdot n}{2^{N}}} + 1.
    \end{equation}
    Further, we can show that
    \begin{align}
        \frac{\rbr{m_{lo}\cdot n} /\; 2^{N} + n}{2^{p}}
         & < \frac{\floor{\rbr{m_{lo}\cdot n} /\; 2^{N}} + n}{2^{p}} + \frac{1}{2^{p}}                                     \\
         & \le \floor{\frac{\floor{\rbr{m_{lo}\cdot n} /\; 2^{N}} + n}{2^{p}}} + \frac{2^{p} - 1}{2^{p}} + \frac{1}{2^{p}} \\
         & = \floor{\frac{\floor{\rbr{m_{lo}\cdot n} /\; 2^{N}} + n}{2^{p}}} + 1.
    \end{align}
    Thus,
    \begin{equation}
        \floor{\frac{\rbr{m_{lo}\cdot n} /\; 2^{N} + n}{2^{p}}} = \floor{\frac{\floor{\rbr{m_{lo}\cdot n} /\; 2^{N}} + n}{2^{p}}} = \floor{\frac{q + n}{2^{p}}}.
        \label{eq:main_target_rhs}
    \end{equation}
    Combining \eqref{eq:main_target_lhs} and \eqref{eq:main_target_rhs}, we prove \eqref{eq:main_target}.
\end{proof}

\subsection{Proof of Proposition \ref{prop:n+q}}

\begin{proof}
    Combining \eqref{eq:main_target} and \eqref{eq:floor-shift}, we know
    \begin{equation}
        \rbr{mn \shr N} \shr p = mn \shr \rbr{N + p} = \floor{\frac{mn}{2^{N+p}}} = \floor{\frac{n + q}{2^p}} = (n + q) \shr p,
    \end{equation}
    which implies that $n + q$ has the same bit width as $mn \shr N$.

    \Cref{rem:65bits} shows that $mn$ takes up to $2N+1$ bits.
    So $n + q$ can take up to $N+1$ bits and overflow an $N$-bit register.
\end{proof}

\subsection{Proof of Proposition \ref{prop:faster-feasible}}

\begin{proof}
    If $n<2^{N-1}$, then $n$ takes only $N-1$ bits at most. Thus, $mn$ takes up to $2N$ bits. Therefore, $n + q$ takes up to $N$ bits, and is feasible for $N$-bit registers.
\end{proof}


\section{Proof of Corollary \ref{thm:no-overflow}}
\label{sec:proof-no-overflow}

\begin{proof}
    Given $d\not=1$, we know $p=\ceil{\log_2{d}}\ge 1$.
    RHS of \eqref{eq:main} can be rewritten as
    \begin{align}
        \floor{\frac{n+q}{2^p}}
        = \floor{\frac{n - q + 2q}{2^p}}
        = \floor{\frac{\rbr{n - q}/\;2 + q}{2^{p-1}}}.
    \end{align}
    Similar to \eqref{eq:main_target_rhs}, we can take floor of $(n-q)/\;2$ and get
    \begin{equation}
        \floor{\frac{n+q}{2^p}} = \floor{\frac{\floor{\rbr{n - q}/\;2} + q}{2^{p-1}}}.
    \end{equation}
    This proves \eqref{eq:main-no-overflow}.
\end{proof}

\subsection{Proof of Proposition \ref{prop:feasible}}

\begin{proof}
    On the one hand, $q=\floor{\rbr{m_{lo}\cdot n}\;/\;{2^{N}}}$ and $m_{lo}<2^{N}$, so $q \le n$ and $n - q$ cannot underflow.
    On the other hand, consider the addition:
    \begin{equation}
        \floor{\frac{n - q}{2}} + q \le {\frac{n - q}{2}} + q = \frac{n + q}{2}.
    \end{equation}
    Rounding down both sides we have
    \begin{equation}
        \floor{\frac{n - q}{2}} + q \le \floor{\frac{n + q}{2}}.
    \end{equation}
    Given that $n + q\in\U{N+1}$, we know $\floor{\rbr{n + q}/\;2}\in \U{N}$.
    Thus the addition cannot overflow.
\end{proof}


\section{Conclusion}

This article provides full proofs of the classic round-up variant of the fast unsigned division by constants algorithm.
Additionally, it presents a simpler variant that applies to strictly bounded dividends with rigorous proof.


\section*{Acknowledgement}

This article is inspired by Yinghan Li's CUDA implementation of the fast 32-bit unsigned division by constants.


\printbibliography


\clearpage


\begin{appendix}

    \section{Bit Width of the Product of Two Unsigned Integers}

    \begin{theorem}
        Let $x\in\U{M}$ and $y\in\U{N}$ with $M > 1$ and $N > 1$, then the bit width of the product $xy$ is \emph{tightly} bounded by $M+N$, i.e.,
        \begin{enumerate}[(1)]
            \item $\forall x\in\U{M}$ and $y\in\U{N}$, $xy\in\U{M+N}$; and
            \item $\exists x\in\U{M}$ and $y\in\U{N}$, such that $xy\not\in\U{M+N-1}$.
        \end{enumerate}
        \label{thm:width}
    \end{theorem}

    \begin{proof}
        Let $\Bar{x}\coloneqq \max\U{M} =2^M - 1$, and $\Bar{y} \coloneqq \max\U{N} = 2^N - 1$.

        First, we will prove proposition (1). $\forall x\in\U{M}$ and $y\in\U{N}$, we have
        \begin{equation}
            xy\le\Bar{x}\Bar{y} = \rbr{2^M-1}\rbr{2^N-1}=2^{M+N} - \rbr{2^M+2^N} + 1 < 2^{M+N},
        \end{equation}
        so by definition $xy\in\U{M+N}$.

        Next, we will prove proposition (2). By definition, $\max\U{M+N-1}=2^{M+N-1}-1$, and we have
        \begin{align}
            \Bar{x}\Bar{y} - \rbr{2^{M+N-1} - 1}
            = 2^{M+N} - \rbr{2^M+2^N} + 1 - {2^{M+N-1}} + 1
            = 2^{M+N-1} - \rbr{2^M+2^N} + 2.
        \end{align}
        To show $\Bar{x}\Bar{y}\not\in\U{M+N-1}$, we only need to show that
        \begin{align}
            2^{M+N-1} - \rbr{2^M+2^N} + 2 > 0
            \iff \frac{1}{2}-\rbr{\frac{1}{2^N}+\frac{1}{2^M}} + \frac{1}{2^{M+N-1}} > 0.
        \end{align}
        Given that $M,N\in\N$, $M>1$, and $N>1$, we have
        \begin{align}
            \frac{1}{2^{M+N-1}}           & > 0,                                             \\
            \frac{1}{2^N} + \frac{1}{2^M} & \le \frac{1}{2^2} + \frac{1}{2^2} = \frac{1}{2}.
        \end{align}
        These show that
        \begin{equation}
            \frac{1}{2}-\rbr{\frac{1}{2^N}+\frac{1}{2^M}} + \frac{1}{2^{M+N-1}} > 0.
        \end{equation}
        Thus, $\exists x\in\U{M}$ and $y\in\U{N}$, such that $xy\not\in\U{M+N-1}$.

        In conclusion, the bit width of $xy$ is tightly bounded by $M+N$.
    \end{proof}

\end{appendix}


\end{document}